\documentclass[submission,copyright,creativecommons]{eptcs}
\providecommand{\event}{TARK 2025} 

\usepackage{iftex}
\usepackage{stmaryrd}
\usepackage{amsmath}
\usepackage{amsthm}
\usepackage{amssymb}
\usepackage{verbatim}

\newtheorem{theorem}{Theorem}[section]

\newtheorem{lemma}{Lemma}[theorem]
\newtheorem{proposition}{Proposition}[section]

\theoremstyle{definition}
\newtheorem{definition}{Definition}[section]
\newtheorem{example}{Example}[section]

\newcommand{\Lbox}[1]{\left[ #1\right] }
\newcommand{\Ldiamond}[1]{\langle #1\rangle}
\newcommand{\Propo}[1]{\llbracket{#1}\rrbracket}

\newcommand{\set}[1]{\{#1\}}

\ifpdf
  \usepackage{underscore}         
  \usepackage[T1]{fontenc}        
\else
  \usepackage{breakurl}           
\fi

\title{Logic of (Common or Distributed)$^\ast$ Knowledge}
\author{Chenwei Shi\thanks{This research is supported by Beijing Philosophy and Social Science Foundation (Grant number: 24DTR012)}
\institute{Department of Philosophy\\
Tsinghua University\\
Beijing, China}
\email{scw@tsinghua.edu.cn}
}

\def\titlerunning{LCDK}
\def\authorrunning{C. Shi}
\begin{document}
\maketitle

\begin{abstract}
  In this paper, we generalize epistemic logic so that it can help reason about ways of combining common knowledge and distributed knowledge such as ``common distributed knowledge'', ``distributed common knowledge'', ``distributed common distributed knowledge'' and so on. Moreover, we study the logic of its dynamic update by arbitrary reading events. We axiomatize these logics and prove their soundness and completeness. 
\end{abstract}

\section{Introduction}

In epistemic logic \cite{reasoning-about-knowledge}, common knowledge and distributed knowledge are two notions of group knowledge which have been studied extensively. 
In a Kripke model $(W,\set{R_a}_{a\in A},V)$ where each binary relation $R_a$ respresents one agent $a$'s epistemic state, 
the distributed knowedge of a group $G\subseteq A$ of agents is represented by $\bigcap_{a\in G} R_a$; 
the common knowledge of $G$ is represented by the transitive relfexive closure of the union of $R_a$ for $a\in G$, i.e. $(\bigcup_{a\in G}R_a)^\ast$, the smallest set of pairs which contains $\bigcup_{a\in G}R_a$ and is closed under transitivity and reflexivity, or equivalently, the set of pairs which can be connected by a path of finite steps of relation  $\bigcup_{a\in G}R_a$.

A new notion of group knowledge, called ``common distributed knowledge'', is studied by a recent paper \cite{AS2023LPAR}.
As the name indicated, this notion combines common knowledge and distributed knowledge somehow. 
In fact, it is a notion of knowledge for a group of groups of agents. 
Taking a group of groups $G = \set{G_1,\ldots,G_n}$, the notion of common distributed knowledge is represented by the relation $(\bigcup^n_{i=1} \bigcap_{a\in G_i} R_a)^\ast$. 
Individual agents' knowledge in the notion of common knowledge is replaced by groups' distributed knowledge, but the definion of common knowledge is kept the same.

The notion ``common distributed knowledge'' is only one way of combining common knowledge and distributed knowledge. 
It indicates other possible ways. 
For example, distributed common knowledge, distributed common distributed knowledge and so on. 
In this paper, we develop a logic which can help reason about all these ways of combining common knowledge and distributed knowledge. 

The multiple iterations of common knowledge and distributed knowledge may sound too complex and abstract to make any sense. 
The following example hopefully makes it clear how these complex epistemic notions can naturally show up.
\begin{example}\label{exap:intro}

    One day, Alice, Bob and Carter are playing card. Their father helps distribute a deck of cards (52) to them. In the current round, Bob and Carter team up, which means that they share one card, in addition to the 25 cards assigned to each of them. They can check the shared card together so that they commonly know what it is. Alice holds the remaining card.
    
    Now according to the rules of the card game they are playing, Alice can choose to check the shared card of Bob and Carter. In return for this, Bob and Carter can check Alice's card. 

    Alice exercises her right and checks the card. Bob and Carter also check Alice's card. In the scenario, the actions of learning are public while what is learned is not (at least to the audience, for example, their father).

    Restricting our attention to the information relevant to the current round of the card game, that is, who has what cards, it is not hard to figure out that Bob and Carter commonly know more than what Alice knows after they check each other's card.
      
    On one hand, what Alice knows after she checked the shared card is a simple instance of ``distributed common knoledge''. It is Alice's common knowledge (for an individual agent, her common knowledge is the same as her knowledge) and Bob and Carter's common knowledge that undergo the operation of distributed knowledge.  On the other hand, what Bob and Carter commonly know after they check Alice's card is a simple instance of ``common distributed knowledge''. It is Bob and Alice's distributed knowledge and Bob and Carter's distributed knowledge that undergo the operation of common knowledge.

    So, in this example,  there is more  common distrubuted knowledge of the group $\{\text{Bob,Alice}\}$ and the group $\{\text{Carter,Alice}\}$ than the distributed common knowledge of Alice and the team $\{\text{Carter,Alice}\}$.  But is it always the case? 
\end{example}
To have a logic for reasoning about the above type of questions, it is key to observe that the intersection of two equivalence relations/reflexive and transitive relations and the reflexive and transtive closure of the union of two equivalence relations/reflexive and transitive relations  can serve as  meet and join respectively in a lattice.
This observation serves as the cornerstone of our axiomatization of the logic and its dynamic extensions. 

In section two, we present the logic where iteration of the operations for defining distributed and common knowledge is allowed. In section three, we extend the logic with a dynamic operator for semi-public reading events. In section four, we consider a more general type of dynamics than semi-public reading events, that is, arbitrary reading events. For the static logic, we provide a sound and complete axiomatization. For both of its dynamic extensions, we find their reduction laws, as usually done in dynamic epistemic logic (c.f. \cite{hansWB2007dynamic} and \cite{sep-dynamic-epistemic}).

\section{Syntax, Semantics and Axiomatization}\label{sec:SSA}

In this section, we present our logic's syntax, semantics and axiom system. 

\begin{definition}[Language $\mathcal{L}$]
    $$\mathsf{T}\ni \tau ::= x\mid \tau +\tau \mid \tau \cdot \tau$$
where $x\in \mathsf{AT}$, a finite set of atomic terms. 
$$\varphi ::= p \mid \neg \varphi\mid \varphi\vee \varphi \mid \Ldiamond{\tau}\varphi$$
where $p\in \mathsf{AF}$, a countable set of atomic formulas, and $\tau \in \mathsf{T}$. $\wedge$, $\rightarrow $ and $\leftrightarrow$ are defined as usual, and $\Lbox{\tau}\varphi := \neg\Ldiamond{\tau}\neg \varphi$.
\end{definition}

The simplest form of distributed knowledge/common knowledge is expressed as $[x_1\cdot x_2]\varphi\,$/$\,[x_1 + x_2]\varphi$ respectively. The common distributed knowledge in  example \ref{exap:intro} is expressed as $\Lbox{(b\cdot a) + (c\cdot a)}\varphi$, while the distributed common knowledge is expressed as $\Lbox{a\cdot (b+c)}\varphi$. It is very hard to read formulas of the form $\Lbox{\tau}\varphi$ in natural language when $\tau$ becomes complex. However, it should be not very hard to see how the formal language can help us express those complex ways of combining distributed and common knowledge in an efficient way. 

\begin{definition}[Regular model]
    A model 
    $\mathbf{M} = (W,\set{R_\tau\mid \tau\in \mathsf{T}},V)$ is tuple where $W$ is a set of possible states, $R_\tau$ is a binary relation on $W$ for each term $\tau$ in $\mathsf{T}$ and $V: \mathsf{AF}\rightarrow \wp(W)$ is a vaulation function. A regular model is a model which satisfies
    \begin{enumerate}
        \item $R_{\sigma\cdot \tau} = R_\sigma \cap R_\tau$ 
        \item $R_{\sigma + \tau} = (R_\sigma\cup R_\tau)^\ast$       
    \end{enumerate}
    for all $\sigma,\tau\in \mathsf{T}$. When all the relations are equivalence relations/reflexive and transitive relations, a regular model is called an ``S5/S4-regular model'' respectively.
\end{definition}
A key observation is that in an S5/S4-regular model, $(\set{R_\tau\subseteq W\times W\mid \tau\in \mathsf{T}},\cdot,+)$ where $R_\sigma\cdot R_\tau = R_\sigma\cap R_\tau$ and $R_\sigma + R_\tau =  (R_\sigma\cup R_\tau)^\ast$ consititutes a lattice. That is, the sublattice of the lattice of all equivalence relations/preorders over $W$ which is generated by $\set{R_x\mid x\in \mathsf{AT}}$.

\begin{definition}[Truth condition]
    For a model $\mathbf{M}$ and a possible state $w$ in it,
    \begin{center}

    \begin{tabular}[h!]{lcl}
        $\mathbf{M},w\models  p$ & iff  & $w\in V(p)$\\
        $\mathbf{M},w\models \neg \varphi$ & iff  & $\mathbf{M},w\not\models \varphi$\\
        $\mathbf{M},w\models \varphi\vee \psi$ & iff  & $\mathbf{M},w\models\varphi$ or $\mathbf{M},w\models\psi$ \\
        $\mathbf{M},w\models \Ldiamond{\tau} \varphi$ & iff  & $\mathbf{M},v\models \varphi$ for some  $v\in R_\tau(w)$\\
        
    \end{tabular}
            
    \end{center}
\end{definition}

\begin{definition}[Axiom system $S5\mathsf{LCDK}$] \mbox{ }\\
    \begin{tabular}[h!]{lll}
        (CPL) & All tautologies in CPL &\\
        (Rules) &  Mpdus Ponens & Necessitation rule \\
        (S5) & & \\
        K & $\Lbox{\tau}(\varphi\rightarrow \psi)\rightarrow (\Lbox{\tau}\varphi\rightarrow \Lbox{\tau}\psi)$ & \\
        T & $\Lbox{\tau}\varphi\rightarrow \varphi$ & \\
        4 & $\Lbox{\tau}\varphi\rightarrow \Lbox{\tau}\Lbox{\tau}\varphi$ & \\
        5 & $\neg\Lbox{\tau}\varphi\rightarrow \Lbox{\tau}\neg\Lbox{\tau}\varphi$ & \\
        dual & $\Ldiamond{\tau}\varphi\leftrightarrow \neg\Lbox{\tau}\neg \varphi$  & \\
        (LATTICE) & & \\
            Idempotency & $\Ldiamond{\tau\cdot \tau}\varphi\leftrightarrow \Ldiamond{\tau}\varphi$ &$\Ldiamond{\tau + \tau}\varphi\leftrightarrow \Ldiamond{\tau}\varphi$\\
            Commutativity & $\Ldiamond{\tau\cdot \sigma}\varphi\leftrightarrow \Ldiamond{\sigma \cdot \tau}\varphi$ & $\Ldiamond{\tau + \sigma}\varphi\leftrightarrow \Ldiamond{\sigma + \tau}\varphi$ \\
            Associativity & $\Ldiamond{(\rho\cdot \sigma)\cdot \tau}\varphi\leftrightarrow \Ldiamond{\rho \cdot (\sigma\cdot \tau)}\varphi$ & $\Ldiamond{(\rho + \sigma) + \tau}\varphi\leftrightarrow \Ldiamond{\rho + (\sigma + \tau)}\varphi$ \\
            Absorption & $\Ldiamond{\tau \cdot (\tau + \sigma)}\varphi\leftrightarrow \Ldiamond{\tau}\varphi$ & $\Ldiamond{\tau + (\tau \cdot \sigma)}\varphi\leftrightarrow \Ldiamond{\tau}\varphi$ \\
        (STAR) & &\\
        FP & \multicolumn{2}{l}{$\Ldiamond{\tau+\sigma}\varphi\leftrightarrow (\varphi\vee (\Ldiamond{\tau}\Ldiamond{\tau+\sigma}\varphi \vee \Ldiamond{\sigma}\Ldiamond{\tau+\sigma}\varphi)) $} \\
        INDUC & \multicolumn{2}{l}{$\Lbox{\tau+\sigma}(\varphi\rightarrow (\Lbox{\tau}\varphi \wedge \Lbox{\sigma}\varphi))\rightarrow (\varphi\rightarrow \Lbox{\tau+\sigma}\varphi)$}\\
    \end{tabular}  
\end{definition}
When the axiom 5 is removed from $S5\mathsf{LCDK}$, the system is called $S4\mathsf{LCDK}$.

We define a preorder over terms based on the above axiom system:
$$\tau\leq \sigma\qquad \text{ iff}_{def}\qquad \vdash \Ldiamond{\tau}\varphi\rightarrow \Ldiamond{\sigma}\varphi$$
which will be quite handy in our proof of the next completeness theorem.
\begin{theorem}\label{thm:completenessStatic}
    $S5\mathsf{LCDK}$ is sound and weakly complete with repsect to the class of S5-regular frames. $S4\mathsf{LCDK}$ is sound and weakly complete with respect to the class of S4-regular frames.
\end{theorem}
The proof makes use of the method of step by step, the details of which can be found in the appendix.

With the help of this sound and complete logic, the question in Exapmle \ref{exap:intro} can be easily answered. To see this, it is enough to realize that the formal counterpart of the question in our logic is whether $\vdash \Lbox{(b+c)\cdot a}\varphi \rightarrow  \Lbox{b\cdot a + c\cdot a}\varphi$ is the case. By basics of lattice theory, $b\cdot a + c\cdot a \leq (b+c)\cdot a$ is the case.  So the answer is yes. 

\section{Semi-public Reading Events}

In this section, we turn to a type of dynamic update proposed in \cite{AS2023LPAR} and demonstrate its natural generalization in our new logic. We also show how Example \ref{exap:intro} can be formalized in this dynamic setting. 

As \cite{AS2023LPAR} puts it
\begin{quote}
    We call all these actions \emph{semi-public `reading' events}. In all of them, some agents get to access (`read') some other agents' knowledge base(s). But \emph{the fact that this access is gained (or not)} is public: it is common knowledge who can `read' whose knowledge base during these events. 
\end{quote}
A semi-public `reading' event is represented by a reading map.
\begin{quote}
    A \emph{reading map} is a function $\alpha: A\rightarrow \wp(A)$, mapping agents $a\in \mathsf{A}$ to sets of agents $\alpha(a)\subseteq \mathsf{AT}$, subject to the constraint that 
    $$a\in \alpha(a)\quad (\text{for every } a\in A).$$
    Intuitively, $\alpha(a)$ is \emph{the set of agents whose information is accessed by} $a$ during this action. So this last constraint means that \emph{every agent $a$ can always re-read her own knowledge base}.     
\end{quote}

In our logic, the reading map can be generalized as follows
$$\beta: \mathsf{AT}\rightarrow \mathsf{T}\enspace .$$
We could have required that 
$$\beta(a) \leq a$$
to make sure that every agent $a$ can always re-read her own knowledge base. However, to stay as general as possible, we refrain from this extra condition.
According to generalized reading maps, agents can read not only other agents' knowledge but also other more complex forms of knowledge, for example, a group's common knowledge.

Given a reading map $\beta$, let $!\beta$ denote a semi-public event  corresponding to it. Given any S5/S4-regular model $\mathbf{M} = (W,\set{R_\tau\mid \tau\in \mathsf{T}},V)$, the event $!\beta$ returns an updated S5/S4-regular model $\mathbf{M}^{!\beta} = (W,\set{R^{!\beta}_\tau\mid \tau\in \mathsf{T}},V)$ which has the same set of possible states $W$, the same valuation function $V$ but has new equivalence relations/preorders $R^{!\beta}_\tau$  given by 
$$R^{!\beta}_\tau:= R_{\beta'(\tau)}\enspace .$$
where $\beta'$ is a function mapping terms in $\mathsf{T}$ to terms in $\mathsf{T}$ which satisfies the following conditions:
\begin{enumerate}
    \item $\beta'(x) = \beta(x)$ for $x\in \mathsf{AT}$;
    \item $\beta'(\tau\cdot\sigma) = \beta'(\tau)\cdot \beta'(\sigma)$;
    \item $\beta'(\tau+\sigma) = \beta'(\tau) + \beta'(\sigma)$. 
\end{enumerate}  
So $\beta'$ is a homomorphism from $\mathsf{T}$ to $\mathsf{T}$ which extends $\beta$ on $\mathsf{AT}$. By this fact, it implies that the updated model is indeed S5/S4-regular.

For each generalized rading map $\beta$, we add dynamic modalities $[!\beta]\varphi$  into the syntax of our logic (resulting in a new language $\mathcal{L}_!$) and evaluate them at a state as follows:
$$\mathbf{M},w\models  [!\beta]\varphi\quad \text{iff}\quad \mathbf{M}^{!\beta},w\models \varphi\enspace .$$

All the above definitions in this section are straightforward generalization of their counterpart in \cite{AS2023LPAR}. They also result in quite straightforward reduction laws:
\begin{itemize}
    \item $\Lbox{!\beta}p \leftrightarrow p$;
    \item $\Lbox{!\beta}\neg\varphi \leftrightarrow \neg\Lbox{!\beta}\varphi$;
    \item $\Lbox{!\beta}(\varphi\wedge \psi)\leftrightarrow \Lbox{!\beta}\varphi\wedge \Lbox{!\beta}\psi$;
    \item $\Lbox{!\beta}\Lbox{\tau}\varphi\leftrightarrow \Lbox{\beta'(\tau)}\Lbox{!\beta}\varphi$.
\end{itemize}

Putting the above reduction laws together with the axiom system $S5\mathsf{LCDK}$/$S4\mathsf{LCDK}$, we get the axiom system for the dynamic logic of semi-public reading, $S5\mathsf{LCDK}!$/$S4\mathsf{LCDK}!$. 
\begin{theorem}
     $S5\mathsf{LCDK}!$/$S4\mathsf{LCDK}!$ is sound and complete with respect to the class of S5/S4-regular frames. 
\end{theorem}
The soundness is easy to check, from which  we get the equal expressivity of $\mathcal{L}$ and $\mathcal{L}_!$ with respect to the class of regular models. So the completeness of $S5\mathsf{LCDK}!$ and $S4\mathsf{LCDK}!$ follows from the completeness of $S5\mathsf{LCDK}$ and $S4\mathsf{LCDK}$ respectively.

In the dynamic logic of semi-public reading, we can formalize the scenario in Example \ref{exap:intro}. Let $a,b,c\in \mathsf{AT}$. Consider a generalized reading map $\beta$ satisfying $\beta(a) = a\cdot (b+c)$ (Alice semi-publicly learns what Bob and Carter commonly know), $\beta(b) = b\cdot a$ (Bob semi-publicly learns what Alice knows) and $\beta(c) = c\cdot a$ (Carter semi-publicly learns what Alice knows). It does not matter what the value of $\beta(x)$ for $x\neq a,b,c$ is. The question is  whether $\Lbox{!\beta}(\Lbox{a}\varphi \rightarrow \Lbox{b+c}\varphi)$ is derivable in $S5\mathsf{LCDK}$ or $S4\mathsf{LCDK}$.
Applying the reduction laws, we can see that $\Lbox{!\beta}(\Lbox{a}\varphi \rightarrow \Lbox{b+c}\varphi)$ is equivalent to 
$$\Lbox{a\cdot (b+c)}\Lbox{!\beta}\varphi \rightarrow \Lbox{b\cdot a+c\cdot a}[!\beta]\varphi$$ 
which brings us back to the end of Section \ref{sec:SSA}.

Speaking more generally, the connection between lattice theory and epistemic logic is key to the work in this paper. In the next section, we demonstrate the benefit of this connection by modeling epistemic update by arbitrary reading events.

\section{Arbitrary reading events}

In a semi-public reading event, an agent's action of reading others' knowledge is public while what she learns is not. In an arbitrary reading event, the action of reading others' knowledge is not necesarily public either.

The uncertainty of an agent about other agents' reading actions can be modelled by epistemic indistinguishability relation between generalized reading maps. Generalizing the definition of \emph{reading  event model} in \cite{AS2023LPAR}, let a \emph{S5 regular reading event model} be a structure $\mathbf{E} = (E,\set{R_\tau\mid \tau\in T},S)$ where $E$ is a finite set of `events', $S: E\rightarrow F(\mathsf{AT},\mathsf{T})$ is a function mapping events $e$ in $E$ to a generlized reading map $S(e)\in  F(\mathsf{AT},\mathsf{T})$. 
$R_\tau$ are equivalence relations on $E$ which satisfy the two conditions of regularity as in a regular model and the following condition:
$$eR_x f \text{ implies } [S(e)](x) \equiv [S(f)](x)\enspace .$$ 
The symbol $\equiv$ is the abbreviation of $\leq \cap \geq$ where $\leq$ is the preorder we define over the set of terms $\mathsf{T}$ at the end of Section \ref{sec:SSA}. 

The epistemic update of a regular model $\mathbf{M} = (W,\set{R^\mathbf{M}_\tau\mid \tau\in \mathsf{T}},V^\mathbf{M})$ by a S5 regular reading event model $\mathbf{E} = (E,\set{R^\mathbf{E}_\tau\mid \tau\in T},S)$ is defined by a straightforward generalization of the product update proposed in \cite{AS2023LPAR}. That is, $\mathbf{M}\otimes \mathbf{E} = (S\times E, \set{R^{\mathbf{M}\otimes \mathbf{E}}_\tau\mid \tau\in \mathsf{T}}, V^{\mathbf{M}\otimes\mathbf{E}})$,  where 
\begin{itemize}
    \item $S\times E = \set{(w,e)\mid w\in W,e\in E}$
    \item $V^{\mathbf{M}\otimes\mathbf{E}} = V^\mathbf{M}$
    \item $R^{\mathbf{M}\otimes \mathbf{E}}_\tau$ is defined recursively :
    \begin{itemize}
        \item $(w,e)R^{\mathbf{M}\otimes \mathbf{E}}_x (w',e')$ iff $wR^\mathbf{M}_{[S(e)](x)} w'$ and $eR^\mathbf{E}_x e'$
        \item $R^{\mathbf{M}\otimes \mathbf{E}}_{\sigma_1\cdot\sigma_2} := R^{\mathbf{M}\otimes \mathbf{E}}_{\sigma_1}\cap R^{\mathbf{M}\otimes \mathbf{E}}_{\sigma_2}$
        \item $R^{\mathbf{M}\otimes \mathbf{E}}_{\sigma_1+\sigma_2} := (R^{\mathbf{M}\otimes \mathbf{E}}_{\sigma_1}\cup R^{\mathbf{M}\otimes \mathbf{E}}_{\sigma_2})^\ast$
    \end{itemize}
\end{itemize}

The following characterization of $R^{\mathbf{M}\otimes \mathbf{E}}_{\tau}$ will play a key role in finding reduction laws and proving their validity.
\begin{proposition}\label{prop:updatedRelation}
    Given an updated model $\mathbf{M}\otimes \mathbf{E} = (W\times E, \set{R^{\mathbf{M}\otimes \mathbf{E}}_\tau\mid \tau\in \mathsf{T}})$ where $\mathbf{E}$ is an S5 regular reading event model and $\mathbf{M}$ is S5/S4 regular model, let $\gamma: E\times \mathsf{T}\rightarrow \mathsf{T}$ be a function defined as follows:
    \begin{itemize}
        \item $\gamma(e,x) = [S(e)](x)$
        \item $\gamma(e,\sigma_1\cdot \sigma_2) = \gamma(e,\sigma_1)\cdot \gamma(e,\sigma_2)$
        \item $\gamma(e,\sigma_1+\sigma_2) = \sum_{eR^\mathbf{E}_{\sigma_1+\sigma_2}f}(\gamma(f,\sigma_1)+\gamma(f,\sigma_2))$
    \end{itemize}
    We have the following characterization of $R^{\mathbf{M}\otimes \mathbf{E}}_{\tau}$ in terms of $R^{\mathbf{M}}_{\gamma(e,\tau)}$ and  $R^{\mathbf{E}}_{\tau}$:
    $$(w,e)R^{\mathbf{M}\otimes \mathbf{E}}_{\tau}(w',e')\quad \text{iff}\quad wR^{\mathbf{M}}_{\gamma(e,\tau)}w' \text{ and } eR^{\mathbf{E}}_{\tau} e'$$
\end{proposition}
\begin{proof}
    
    We prove by induction on the structure of $\tau$. The basic case follows directly from the definition of $R^{\mathbf{M}\otimes \mathbf{E}}_{x}$. The case for $\sigma_1\cdot \sigma_2$ follows directly from the definition of $R^{\mathbf{M}\otimes \mathbf{E}}_{\sigma_1\cdot\sigma_2}$, the inductive hypothesis and the regularity of $\mathbf{M}$ and $\mathbf{E}$.

    Now consider the case where $\tau$ is of the form $\sigma_1+\sigma_2$. 
    
    Assume that  $(w,e)R^{\mathbf{M}\otimes \mathbf{E}}_{\sigma_1+\sigma_2}(w',e')$. There is a path 
    $$(w,e)R^{\mathbf{M}\otimes \mathbf{E}}_{\tau_1}(w_1,e_1)R^{\mathbf{M}\otimes \mathbf{E}}_{\tau_2}\ldots R^{\mathbf{M}\otimes \mathbf{E}}_{\tau_n}(w_n,e_n)R^{\mathbf{M}\otimes \mathbf{E}}_{\tau_{n+1}}(w',e')\enspace .$$ 
    Each $R^{\mathbf{M}\otimes \mathbf{E}}_{\tau_i}$ in the path is either $R^{\mathbf{M}\otimes \mathbf{E}}_{\sigma_1}$ or $R^{\mathbf{M}\otimes \mathbf{E}}_{\sigma_2}$. So it follows from inductive hypothesis that 
    $$eR^{\mathbf{E}}_{\tau_1}e_1R^{\mathbf{E}}_{\tau_2}\ldots R^{\mathbf{E}}_{\tau_n}e_nR^{\mathbf{E}}_{\tau_{n+1}}e'\quad \text{and}\quad  wR^{\mathbf{M}}_{\gamma(e,\tau_1)}w_1R^{\mathbf{M}}_{\gamma(e_1,\tau_2)}\ldots R^{\mathbf{M}}_{\gamma(e_{n-1},\tau_n)}w_nR^{\mathbf{M}}_{\gamma(e_n,\tau_{n+1})}w'\enspace. $$ 
    Each $\tau_i$ in $R^{\mathbf{M}\otimes \mathbf{E}}_{\tau_i}$ is either $\sigma_1$ or $\sigma_2$. So $e R^{\mathbf{E}}_{\sigma_1+\sigma_2} e'$ and $w R^\mathbf{M}_{\sum_{eR^\mathbf{E}_{\sigma_1+\sigma_2}f}(\gamma(f,\sigma_1)+\gamma(f,\sigma_2))}w'$.

    Assume that $e R^{\mathbf{E}}_{\sigma_1+\sigma_2} e'$ and $w R^\mathbf{M}_{\sum_{eR^\mathbf{E}_{\sigma_1+\sigma_2}f}(\gamma(f,\sigma_1)+\gamma(f,\sigma_2))}w'$. So there is a path 
    $$\bigstar \qquad wR^{\mathbf{M}}_{\gamma(f_0,\tau_1)}w_1R^{\mathbf{M}}_{\gamma(f_1,\tau_2)}\ldots R^{\mathbf{M}}_{\gamma(f_{n-1},\tau_n)}w_nR^{\mathbf{M}}_{\gamma(f_n,\tau_{n+1})}w'\enspace. $$ 
    Since $eR^\mathbf{E}_{\sigma_1+\sigma_2}f_i$ for all $f_i$ in the path, $e R^{\mathbf{E}}_{\sigma_1+\sigma_2} e'$ and $R^\mathbf{E}_{\sigma_1+\sigma_2}$ is an equivalence relation, there is a path from $f_i$ to $f_{i+1}$ (when $i=n$, $f_{i+1} = e'$):
    $$f_iR^\mathbf{E}_{\tau_{i+1}} f_i R^\mathbf{E}_{\sigma^i_1} e^i_1 R^\mathbf{E}_{\sigma^i_2}\ldots e^i_{m(i)}R^\mathbf{E}_{\sigma^i_{m(i)+1}}f_{i+1}\enspace .$$

    Each  step is $R^\mathbf{E}_{\sigma^i_k}$ where $\sigma^i_k$ is either $\sigma_1$ or $\sigma_2$.
    Now, we can expand each step $w_iR^{\mathbf{M}}_{\gamma(f_i,\tau_{i+1})}w_{i+1}$ in the path $\bigstar$
    to the following one (when $i=n$, $w_{i+1} = w'$):
    $$w_iR^{\mathbf{M}}_{\gamma(f_i,\tau_{i+1})}w_{i+1} R^{\mathbf{M}}_{\gamma(f_i,\sigma^i_1)}w_{i+1}R^{\mathbf{M}}_{\gamma(e^i_1,\sigma^i_2)}w_{i+1}\ldots w_{i+1}R^{\mathbf{M}}_{\gamma(e^i_{m(i)},\sigma^i_{m(i)+1})}w_{i+1} \enspace .$$ 
    So we get a new path from $w$ to $w'$. Each step is $R^{\mathbf{M}}_{\gamma(e^i_k,\sigma^i_k)}$ or $R^{\mathbf{M}}_{\gamma(f_i,\tau_{i+1})}$ where $\sigma^i_k$ and $\tau_{i+1}$ are either $\sigma_1$ or $\sigma_2$. 
    By the inductive hypothesis, we have
    $$(w_i,f_i)R^{\mathbf{M}\otimes \mathbf{E}}_{\tau_{i+1}} (w_{i+1},f_i) R^{\mathbf{M}\otimes \mathbf{E}}_{\sigma^i_1} (w_{i+1},e^i_1)R^{\mathbf{M}\otimes \mathbf{E}}_{\sigma^i_2}(w_{i+1},e^i_2)\ldots R^{\mathbf{M}\otimes \mathbf{E}}_{\sigma^i_{m(i)}}(w_{i+1},e^i_{m(i)})R^{\mathbf{M}\otimes \mathbf{E}}_{\sigma^i_{m(i)+1}}(w_{i+1},f_{i+1})$$   
    which implies that $(w_i,f_i)R^{\mathbf{M}\otimes \mathbf{E}}_{\sigma_1+\sigma_2}(w_{i+1},f_{i+i})$
    by the definition of $R^{\mathbf{M}\otimes \mathbf{E}}_{\sigma_1+\sigma_2}$. Therefore, we have
    $$(w,f_0)R^{\mathbf{M}\otimes \mathbf{E}}_{\sigma_1+\sigma_2}(w_1,f_1)R^{\mathbf{M}\otimes \mathbf{E}}_{\sigma_1+\sigma_2}\ldots R^{\mathbf{M}\otimes \mathbf{E}}_{\sigma_1+\sigma_2}(w_n,f_n)R^{\mathbf{M}\otimes \mathbf{E}}_{\sigma_1+\sigma_2}(w',e')\enspace .$$    
    Note that $f_0$ is not necessarily $e$. However, there is a path from $e$ to $f_0$ in $\mathbf{E}$ via either $R^\mathbf{E}_{\sigma_1}$ or $R^\mathbf{E}_{\sigma_2}$:
    $$e R^\mathbf{E}_{\sigma^\ast_1} e^\ast_1 R^\mathbf{E}_{\sigma^\ast_2}\ldots e^\ast_{m(\ast)}R^\mathbf{E}_{\sigma^\ast_{m(\ast)+1}}f_0\enspace $$
    and a path from $w$ to $w$   
    $$wR^{\mathbf{M}}_{\gamma(e,\sigma^\ast_1)}wR^{\mathbf{M}}_{\gamma(e^\ast_1,\sigma_{\ast_2})}w\ldots wR^{\mathbf{M}}_{\gamma(e^\ast_{m(\ast)},\sigma^\ast_{m(\ast)+1})}w \enspace. $$ 
    So by inductive hypothesis and the definition of $R^{\mathbf{M}\otimes \mathbf{E}}_{\sigma_1+\sigma_2}$, we have 
    $$(w,e)R^{\mathbf{M}\otimes \mathbf{E}}_{\sigma_1+\sigma_2}(w,f_0)\enspace .$$
    Therefore, $(w,e)R^{\mathbf{M}\otimes \mathbf{E}}_{\sigma_1+\sigma_2}(w',e')$.
\end{proof}

To the language $\mathcal{L}$, we add dynamic operators $\Lbox{\mathbf{E},e}\varphi$ where $\mathbf{E}$ is a finite S5 regular reading event model and $e\in \mathbf{E}$. The new language is denoted by $\mathcal{L}_{\mathfrak{E}5}$ where $\mathfrak{E}5$ denotes the class of all finite S5-regular reading event models. The truth condition of $\Lbox{\mathbf{E},e}\varphi$ in a regular model $\mathbf{M}$ is given by 
$$\mathbf{M},w\models \Lbox{\mathbf{E},e}\varphi\quad \text{iff}\quad \mathbf{M}\otimes \mathbf{E},(w,e)\models \varphi\enspace .$$

The reduction laws for $\mathcal{L}_\mathfrak{E}$ are the following ones:
\begin{itemize}
    \item $\Lbox{\mathbf{E},e}p\leftrightarrow p$
    \item $\Lbox{\mathbf{E},e}\neg\varphi\leftrightarrow \neg\Lbox{\mathbf{E},e}\varphi$
    \item $\Lbox{\mathbf{E},e}(\varphi\wedge\psi)\leftrightarrow (\Lbox{\mathbf{E},e}\varphi\wedge \Lbox{\mathbf{E},e}\psi)$
    \item $\Lbox{\mathbf{E},e}\Lbox{\tau}\varphi\leftrightarrow \bigwedge_{eR^\mathbf{E}_\tau e'}\Lbox{\gamma(e,\tau)}\Lbox{\mathbf{E},e'}\varphi$
\end{itemize}
Where $\gamma$ is the function defined in Proposition \ref{prop:updatedRelation}.
\begin{proposition}
    All above reduction axioms are valid in the class of S4/S5-regular frames.
\end{proposition}
\begin{proof}
    We only prove the validity of $\Lbox{\mathbf{E},e}\Lbox{\tau}\varphi\leftrightarrow \bigwedge\set{\Lbox{\gamma(e,\tau)}\Lbox{\mathbf{E},e'}\varphi\mid eR^\mathbf{E}_\tau e'}$. The validity of other reduction laws are trivial.

    Take an arbitrary regular model $\mathbf{M}$ and an arbitrary world $w$ in it. 

    $\mathbf{M},w\models \Lbox{\mathbf{E},e}\Lbox{\tau}\varphi$ iff $\mathbf{M}\otimes\mathbf{E},(w,e)\models \Lbox{\tau}\varphi$ iff $\forall (w',e')\in \mathbf{M}\otimes\mathbf{E}: (w,e)R^{\mathbf{M}\otimes\mathbf{E}}_\tau (w',e')\Rightarrow (w',e')\models \varphi$ iff (by Proposition \ref{prop:updatedRelation}) $\forall w'\in \mathbf{M},e'\in \mathbf{E}: wR^\mathbf{M}_{\gamma(e,\tau)} w'\, \&\, eR^\mathbf{E}_\tau e' \Rightarrow \mathbf{M}\otimes\mathbf{E},(w',e')\models \varphi$ iff $\forall w'\in \mathbf{M},e'\in \mathbf{E}: wR^\mathbf{M}_{\gamma(e,\tau)} w'\, \&\, eR^\mathbf{E}_\tau e' \Rightarrow \mathbf{M},w'\models \Lbox{\mathbf{E},e'}\varphi$ iff  $\forall e'\in \mathbf{E}: eR^\mathbf{E}_\tau e' \Rightarrow \forall w'\in \mathbf{M}\, (wR^\mathbf{M}_{\gamma(e,\tau)} w' \Rightarrow \mathbf{M},w'\models \Lbox{\mathbf{E},e'}\varphi)$ iff $\forall e'\in \mathbf{E}: eR^\mathbf{E}_\tau e' \Rightarrow \mathbf{M},w\models \Lbox{\gamma(e,\tau)}\Lbox{\mathbf{E},e'}\varphi$ iff $\mathbf{M},w\models \bigwedge_{eR^\mathbf{E}_\tau e'}\Lbox{\gamma(e,\tau)}\Lbox{\mathbf{E},e'}\varphi$
\end{proof}

Putting the above reduction laws together with the axiom system $S5\mathsf{LCDK}$/$S4\mathsf{LCDK}$, we get the axiom system for the dynamic logic of arbitrary reading, $S5\mathsf{LCDK}\mathfrak{E}$/$S4\mathsf{LCDK}\mathfrak{E}$. 
\begin{theorem}
     $S5\mathsf{LCDK}\mathfrak{E}$/$S4\mathsf{LCDK}\mathfrak{E}$ is sound and complete with respect to the class of S5/S4-regular frames. 
\end{theorem}
The soundness follows from the validity of the reduction laws. By the validity of the reduction laws,  we also get the equal expressivity of $\mathcal{L}$ and $\mathcal{L}_{\mathfrak{E}5}$ with respect to the class of regular models. So the completeness of $S5\mathsf{LCDK}\mathfrak{E}$ and $S4\mathsf{LCDK}\mathfrak{E}$ follows from the completeness of $S5\mathsf{LCDK}$ and $S4\mathsf{LCDK}$ respectively.

\section{Conclusion and Future Work}

By viewing distributed knowledge and common knowledge as meet and join in a lattice respectively, we have shown that the epistemic logic with operators for distributed knowledge and common knowledge can be extended to express more involved types of group knowledge in the S4 and S5 cases. To characterize the epistemic update by arbitrary reading events in the extended logic $\mathsf{LCDK}$, we provide its reduction laws.  

The connection with lattice theory is key to the extension proposed in this paper. By attaching the epistemic interpretation to lattice theoretical notions, it becomes possible to deepen our understanding of epistemic notions by exploring lattice theoretical results. For example, what if we ask the lattice of distributed knowledge and common knowledge to be modular or even distributive? Some results in lattice theory can be found in \cite{Ore1942}.

Although this paper focuses on the epistemic interpretation of the lattice structre of the equivalence relations/reflexive and transitive relations, there are other possible interpretations. For example, questions in \cite{Minica2012}; functional dependence in \cite{Baltag2021}; priority structures in \cite{CHRISTOFF-GRATZL-ROY-2022}, and so on.  The lattice theoretical perspective of this paper can be applied to the other interpretations too.

\appendix
\section{Soundness and Completeness}

Soundness is easy to see, so we ignore the proof. The proof for weak completeness comes next. 

Steps (1) to (4) are adapted from Section 4.8 of \cite{Blackburn2001}.

\subsection{Step (1) Fischer-Ladner closure and atoms}

\begin{definition}[Fischer-Ladner closure]

    Let $X$ be a set of formulas. Then $X$ is FL-closed if it is closed under subformulas and satisfies the following additional constraints
    \begin{enumerate}
        \item If $\Ldiamond{\tau\cdot\sigma}\varphi\in X$, then $\Ldiamond{\tau}\varphi\wedge \Ldiamond{\sigma}\varphi\in X$;
        \item If $\Ldiamond{\tau +\sigma}\varphi\in X$, then $\Ldiamond{\tau}\Ldiamond{\tau +\sigma}\varphi, \Ldiamond{\sigma}\Ldiamond{\tau + \sigma}\varphi\in X$.
    \end{enumerate}
    If $\Sigma$ is any set of formulas, then $\mathrm{FL}(\Sigma)$ is the smallest set of formulas containing $\Sigma$ that is $FL$ closed.

\end{definition}

\begin{definition}
    We define $\neg \mathrm{FL}(\Sigma)$, the closure of $\Sigma$, as the smallest set containing $\Sigma$ which is FL-closed and closed under single negations. The single negation of a formulas $\phi$ is 
    $$\sim \phi := 
    \begin{cases}
        \psi &  \phi \text{ is of the form } \neg \psi\\
        \neg \phi & \text{ otherwise}
    \end{cases}
    $$
\end{definition}

\begin{definition}[Atoms]
    Let $\Sigma$ be a set of formulas. A set of formulas $A$ is an atom over $\Sigma$ if it is a maxiaml consistent subset of $\neg \mathrm{FL}(\Sigma)$. 
\end{definition}

\begin{lemma}
    Let $\Sigma$ be any set of formulas, and $A$ any element of $At(\Sigma)$. Then 
    \begin{itemize}
        \item For all $\phi$, exactly one of $\phi$ and $\sim\phi$ is in $A$.
        \item For all $\phi\vee \psi\in \neg \mathrm{FL}(\Sigma)$: $\phi\vee \psi\in A$ iff $\phi\in A$ or $\psi\in A$.
        \item For all $\Ldiamond{\tau}\phi, \Ldiamond{\sigma}\phi\in \neg \mathrm{FL}(\Sigma)$: if $\sigma\leq \tau$ and  $\Ldiamond{\sigma}\phi\in A$, then $\Ldiamond{\tau}\phi\in A$.
        \item For all $\Ldiamond{\tau+\sigma}\phi \in \neg \mathrm{FL}(\Sigma)$: $\Ldiamond{\tau+\sigma}\phi \in A$ iff $\phi\in A$ or $\Ldiamond{\tau}\Ldiamond{\tau+\sigma}\phi\in A$ or $\Ldiamond{\sigma}\Ldiamond{\tau+\sigma}\phi\in A$.
    \end{itemize}
\end{lemma}

\begin{lemma}
    Let $\mathcal{M}$ be the set of all MCSs, and $\Sigma$ a set of formulas. 
    $At(\Sigma) = \set{\Gamma\cap \neg\mathrm{FL}(\Sigma)\mid \Gamma\in \mathcal{M}}$.
\end{lemma}

\begin{lemma}
If $\phi\in \neg \mathrm{FL}(\Sigma)$ and $\phi$ is consistent, then there is an $A\in At(\Sigma)$ such that $\phi\in A$.    
\end{lemma}

\subsection{Step (2): Finite Canonical Model over $\Sigma$}

\begin{definition}

    Let $\Sigma$ be a finite set of formulas. The finite canonical model (FCM) over $\Sigma$ is the triple 
    $\mathcal{M}_F = (At(\Sigma), \set{S^\Sigma_\tau}_{\tau\in T},V^\Sigma)$
    where for all propositional variables $p$, $V^\Sigma(p) = \set{A\in At(\Sigma)\mid p\in A}$, and for all atoms $A,B\in At(\Sigma)$ and for all terms $\tau$, 
    $AS^\Sigma_\tau B \text{ if } \hat{A}\wedge \Ldiamond{\tau} \hat{B} \text{ is consistent}$.
    
\end{definition}

\begin{lemma}\label{lem:reflexSymm}
    The canonical relations are reflexive and symmetric.
\end{lemma}

The above lemma can also be proved as a corollary to the following fact. 
\begin{proposition}\label{prop:inherit} 
For $A,B\in At(\Sigma)$,
    $ AS_\tau B \text{ iff there is } a\in \Propo{\hat{A}}_\mathcal{M} \text{ and there is } b\in \Propo{\hat{B}}_\mathcal{M} \text{ such that } a R^\mathcal{M}_\tau b $
where $\mathcal{M}$ is the canonical model.
\end{proposition}
This fact also indicates that $S_\tau$ is not necessarily transitive. Moreover, it makes the following existence lemma a corollary to the existence lemma for the canonical model.

\begin{lemma}[Existence lemma for FCM]
    Let $A$ be an atom, and let $\tau$ be a term in $T$. Then for all formulas $\Ldiamond{\tau}\psi$ in $\neg \mathrm{FL}(\Sigma)$, $\Ldiamond{\tau}\psi\in A$ iff there is a $B\in At(\Sigma)$ such that $AS_\tau B$ and $\psi\in B$.
\end{lemma}

\begin{lemma}\label{lem:blue4.87}
    $S_{\tau+\sigma}\subseteq (S_\tau\cup S_\sigma)^\ast$. 
\end{lemma}

\begin{lemma}
    $S^\Sigma_{\sigma\cdot\tau}\subseteq S^\Sigma_\sigma\cap S^\Sigma_\tau$.
\end{lemma}
\subsection{Step (3): Transitive FCM}
\begin{definition}[Transitive FCM (TFCM)]

    Given a FCM $\mathcal{M}^\Sigma_F = (At(\Sigma), \set{S^\Sigma_\tau}_{\tau\in T},V^\Sigma)$, we define its transtive closure as 
    $\mathcal{M}^\Sigma_{TF} = (At(\Sigma), \set{(S^\Sigma_\tau)^t}_{\tau\in T},V^\Sigma)$
    where $(S^\Sigma_\tau)^t$ is the transitive closure of $S^\Sigma_\tau$.
    
\end{definition}

To prove the existence lemma for TFCM,  the direction from left to right immediately follows from the existence lemma for FCM, because $S^\Sigma_\tau\subseteq (S^\Sigma_\tau)^t$. The other direction follows from the following fact.
\begin{proposition}\label{prop:AtauEBtau}
    For $A,B\in At(\Sigma)$,
    $AS^\Sigma_\tau B \text{ implies that } A_\tau = B_\tau$
    where $A_\tau = \set{\varphi\mid \Ldiamond{\tau}\varphi\in A}$.
\end{proposition}

\begin{lemma}[Existence lemma for TFCM]\label{lem:existenceTFCM}
    Let $A$ be an atom, and let $\tau$ be a term in $T$. Then for all formulas $\Ldiamond{\tau}\psi$ in $\neg \mathrm{FL}(\Sigma)$, $\Ldiamond{\tau}\psi\in A$ iff there is a $B\in At(\Sigma)$ such that $A(S_\tau)^t B$ and $\psi\in B$.
\end{lemma}

\subsection{Step (4): Quasi-regularity of TFCM}

\begin{lemma} $(S_{\tau+\sigma})^t = (S_\tau\cup S_\sigma)^\ast$.
\end{lemma}

By the above lemma, TFCM only gets one of the two constraints for a regular model right. In this sense, it is quasi-regular. The other constraint $R_{\sigma\cdot \tau} = R_\sigma\cap R_\tau$ does not hold. In fact, only the direction from right to left does not necessarily hold in a TFCM.
\begin{proposition}
    $(S^\Sigma_{\sigma\cdot\tau})^t\subseteq (S^\Sigma_\sigma)^t\cap (S^\Sigma_\tau)^t$
\end{proposition}

\subsection{Step (5): Building S5-Regular Model Step by Step}\label{sec:stepBystep}

\begin{definition}
A network is a triple $\mathcal{N} = (N,\set{R_\tau}_{\tau\in T},\nu)$ such that $R_\tau$ is a binary relation on a set $N$, and $\nu$ is a labeling function mapping each point in $N$ to an atom in $At(\Sigma)$ for a finite set of formulas $\Sigma$.
\end{definition}

\begin{definition}[Walk and Path in a network]
    Given a network $\mathcal{N}$, a walk from $u$ to $v$ is a sequence  $(w_0, R_1\ldots,R_n w_n)$ satisfying $w_0 = u$, $w_n = v$ and $w_iR_i w_{i+1}$ in $\mathcal{N}$; a path from $u$ to $v$ is a walk where $w_i\neq w_j$ for $i\neq j$. A $\tau$-walk from $u$ to $v$ is a walk $(w_0, R_\tau\ldots,R_\tau w_n)$; a $\tau$-path is a $\tau$-walk from $u$ to $v$ where $w_i\neq w_j$ for $i\neq j$. We stipulate that $(u)$ is a $\tau$-path for any $\tau\in T$.
\end{definition}
We will also talk about notions of walk and path in other strctures, the definitions of which are analogous.

\begin{definition}
    A network  $\mathcal{N} = (N,\set{R_\tau}_{\tau\in T},\nu)$ is coherent if for any  $u,v\in \mathcal{N}$ and any $\tau\in T$, 
    \begin{description}
        \item[(C0)] $R_\tau$ is symmetric;
        \item[(C1)] if $uR_\tau v$, then $\widehat{\nu (u)}\wedge \Ldiamond{\tau}\widehat{\nu (v)}$ is consistent;
        \item[(C2)] if $uR_\tau v$ and $\tau\leq \tau'$, then $uR_{\tau'} v$;
        \item[(C3)] if there is a path from $u$ to $v$ where $u\neq v$, then there is $\sigma\in T$ such that 
        $$\uparrow \sigma = \set{\tau \in T\mid \text{ there is a $\tau$-walk from $u$ to $v$}}$$
        the corresponding relation $R_\sigma$ is called the foundational bridge between $u$ and $v$ in $\mathcal{N}$.\footnote{
            Strictly speaking, there are more than one foundational bridges from $u$ to $v$. However, for any two foudational bridges from $u$ to $v$ in a coherent network, $R_\sigma$ and $R_\tau$, $R_\sigma = R_\tau$ because of condition (C2).}
        
    \end{description}
\end{definition}

Here are some basic facts about walks, paths and foundational bridges in a coherent network.
\begin{proposition}\label{prop:pathFact}
    Given a coherent network $\mathcal{N} = (N,\set{R_\tau}_{\tau\in T},\nu)$
    \begin{enumerate}
        \item $\set{\tau \in T\mid \text{ there is a $\tau$-walk from $u$ to $v$}} = \set{\tau \in T\mid \text{ there is a $\tau$-path from $u$ to $v$}}$
        \item if $uR_{\sigma_0}w_1R_{\sigma_1}\ldots R_{\sigma_{n-1}}w_nR_{\sigma_n}v$ is a path from $u$ to $v$,  then $uR_{\sum^n_0\sigma_i}w_1R_{\sum^n_0\sigma_i}\ldots R_{\sum^n_0\sigma_i}w_nR_{\sum^n_0\sigma_i}v$ is a path. 
        \item If $R_\tau$ is the basic bridge between $u$ and $u'$, $R_\sigma$ is the basic bridge between $u'$ and $u''$ and all paths from $u$ to $u''$ go via $u'$, then $R_{\tau+\sigma}$ is the basic bridge between $u$ and $u''$.
    \end{enumerate}

\end{proposition}
\begin{proof}
    (1) $\set{\tau \in T\mid \text{ there is a $\tau$-walk from $u$ to $v$}} \supseteq \set{\tau \in T\mid \text{ there is a $\tau$-path from $u$ to $v$}}$ is obvious.

    The other direction holds because from any $\tau$-walk from $u$ to $v$, we can get a $\tau$-path by pruning the $\tau$-walk: traversing the $\tau$-walk from $u$, every time a node $w$ in the $\tau$-walk appears twice, deleting the nodes between the first and the second appearances of $w$ including the second appearance of $w$.    

    (2) follows from the condition C2.

    (3) Take any path from $u$ to $u''$, $uR_\delta w_1 R_\delta\ldots R_\delta w_nR_\delta u''$. It goes via $u'$. So there is $i$ such that $s_i = u'$. $R_\tau$ is the foundational bridge from $u$ to $u'$ and $R_\sigma$ is the foundational bridge between $u'$ and $u''$ . So  $\tau\leq \delta$ and $\sigma\leq \delta$, which imply that $\tau+\sigma\leq \delta$. So $\delta\in \uparrow (\tau+\sigma)$.  
    
    On the other hand, there is a $\tau$-path $uR_\tau w_1 R_\tau\ldots R_\tau w_nR_\tau u'$ and  a $\sigma$-path $u'R_\sigma v_1 R_\sigma\ldots R_\sigma v_mR_\tau u''$. It implies that $uR_\tau w_1 R_{\tau+\sigma}\ldots R_{\tau+\sigma} w_nR_{\tau+\sigma} u'R_{\tau+\sigma} v_1 R_{\tau+\sigma}\ldots R_{\tau+\sigma} v_mR_\tau u''$ is a $\tau+\sigma$-path.
    
    Therefore, $R_{\tau+\sigma}$ is the foundational bridge from $u$ to $u''$.
\end{proof}

\begin{definition}
    A network  $\mathcal{N} = (N,\set{R_\tau}_{\tau\in T},\nu)$ is saturated if it satisfies: 
    \begin{description}
        \item[(S1)] for any  $u,v\in \mathcal{N}$ and any $\sigma+\sigma'\in T$, if $uR_{\sigma+\sigma'}v$, then there is a path from $u$ to $v$ each step of which is $R_\sigma$ or $R_{\sigma'}$.
        \item[(S2)] $\mathcal{N}$ is modally saturated. That is, we demand that if $\Ldiamond{\tau}\psi\in \nu(u)$ for some $u\in N$, then there is some $v\in N$ such that $uR_\tau v$ and $\psi\in \nu(v)$.
    \end{description}
\end{definition}

\begin{definition}
    Let $\mathcal{N} = (N,\set{R_\tau}_{\tau\in T},\nu)$ be a network. The frame $\mathfrak{F}_\mathcal{N} = (N,\set{R_\tau}_{\tau\in T})$ is called the underlying frame of $\mathcal{N}$. The induced valuation $V_\mathcal{N}$ on $\mathfrak{F}_\mathcal{N}$ is defined by $V_\mathcal{N}(p) = \set{s\in N\mid p\in \nu(s)}$ for $p\in \neg FL(\Sigma)$. The structure $\mathfrak{I}_\mathcal{N} = (\mathfrak{F},V_\mathcal{N})$ is the induced model.
\end{definition}

A network is good if it is both coherent and saturated.
\begin{lemma}[Existence Lemma for good network]
    Let $\mathcal{N}$ be a good network. For all formulas $\Ldiamond{\tau}\psi$ in $\neg \mathrm{FL}(\Sigma)$ and $s\in N$, $\Ldiamond{\tau}\psi\in \nu(s)$ iff there is a $s'\in N$ such that $sR_\tau s'$ and $\psi\in \nu(s')$.
\end{lemma}

\begin{proposition}
    A good network $\mathcal{N} = (N,\set{R_\tau}_{\tau\in T},\nu)$ satsifies: for any $\tau,\tau',\sigma+\sigma'\in T$ and $u,v\in N$
    \begin{enumerate}
        \item $R_{\tau\cdot\tau'} \subseteq  R_\tau\cap R_{\tau'}$;
        \item $R_{\sigma+\sigma'} \subseteq (R_\sigma\cup R_{\sigma'})^\ast$
    \end{enumerate}
\end{proposition}
\begin{proof}
    The first follows from (C2) and the second follows from (S1).
\end{proof}

\begin{definition}
    Given a good network $\mathcal{N} = (N,\set{R_\tau}_{\tau\in T},\nu)$, let the reflexive and transitive closure of $R_\tau$ for each $\tau\in T$ be denoted by $E_\tau$, the perfect network generated from $\mathcal{N}$ is a triple $\mathcal{P}_\mathcal{N} = (N,\set{E_\tau}_{\tau\in T},\nu)$. 
\end{definition}

\begin{proposition}
    A perfect network satsisfies C0, C2, C3, S1 and S2. Moreover, the frame underlying a perfect network is S5-regular.
\end{proposition}

\begin{proof}
    Given a perfect network $\mathcal{P}_\mathcal{N} = (N,\set{E_\tau}_{\tau\in T},\nu)$ from the good network $\mathcal{N} = (N,\set{R_\tau}_{\tau\in T},\nu)$, we first prove that it satisfies C0, C2, C3, S1 and S2.

    (C0) Assume that $uE_\tau v$, then $uR_\tau v$ or there is $w_1,\ldots,w_n\in N$ such that $uR_\tau w_1R_\tau\ldots R_\tau w_nR_\tau v$. $R_\tau$ is symmetric, so $v E_\tau u$. 


    (C2) Assume that $uE_\tau v$ and $\tau\leq \tau'$. If $uR_\tau v$, then it follows that $uR_{\tau'} v$ and thus $uE_{\tau'} v$. So assume that $uR_\tau v$ is not the case. Then there is $w_1,\ldots,w_n\in N$ such that $uR_\tau w_1R_\tau\ldots R_\tau w_nR_\tau v$. So $uR_{\tau'} w_1 R_{\tau'}\ldots R_{\tau'} w_n R_{\tau'} v$ and thus  $uE_{\tau'}  v$.

    (C3) Assume that there is a path from $u$ to $v$ in $\mathcal{P}_\mathcal{N}$. It follows that there is a path in $\mathcal{N}$ from $u$ to $v$. So there is $\sigma\in T$ such that $\sigma$ is a foundational bridge from $u$ to $v$.
    It is easy to see that 
    $$\uparrow \sigma \subseteq \set{\tau \in T\mid \text{ there is a $\tau$-path from $u$ to $v$ in }\mathcal{P}_\mathcal{N}}.$$

    For the other direction, take any path $uE_\tau w_1 E_\tau\ldots E_\tau w_n E_\tau v$, then there must be a path between $w_i$ and $w_{i+1}$ in $\mathcal{N}$ for any $1\leq i\leq n-1$:
    $w_iR_\tau w^1_i R_\tau w^2_i R_\tau\ldots R_\tau w^{m_i}_iR_\tau w_{i+1}$.
    So there is a $\tau$-path from $u$ to $v$ in $\mathcal{N}$, which implies that $\tau\in \uparrow \sigma$. This implies that 
    $$\uparrow \sigma \supseteq \set{\tau \in T\mid \text{ there is a $\tau$-path from $u$ to $v$ in }\mathcal{P}_\mathcal{N}}.$$
    
    (S1) Assume that $uE_{\sigma+\sigma'} v$. It follows that $u=w_0R_{\sigma+\sigma'} w_1 R_{\sigma+\sigma'}\ldots R_{\sigma+\sigma'} w_n R_{\sigma+\sigma'} w_{n+1}=v$ for some $w_1,\ldots,w_n$. For each $w_i R_{\sigma+\sigma'} w_{i+1}$, there is a path 
    $w_iR_{\sigma^1_i}w^1_i R_{\sigma^2_i}w^2_i\ldots w^{m_i-1}_iR_{\sigma^{m_i}_i} w_{i+1}$
    where $R_{\sigma^j_i}$ is $R_\sigma$ or $R_{\sigma'}$. 
    So there is a path from $u$ to $v$ each step of which is $R_\sigma$ or $R_{\sigma'}$, which implies that there is a path from $u$ to $v$ each of which is $E_\sigma$ or $E_{\sigma'}$.

    (S2) This condition obviously holds in $\mathcal{P}_\mathcal{N}$

    $E_{\tau\cdot\tau'} \subseteq  E_\tau\cap E_{\tau'}$: This follows from C2.

    $E_{\tau\cdot\tau'} \supseteq  E_\tau\cap E_{\tau'}$: Assume that $uE_\tau v$ and $uE_{\tau'}v$. By C3, there is a foundational bridge $E_\sigma$ between $u$ and $v$. So for some $w_1,\ldots,w_n\in N$, $uE_\sigma w_1 E_\sigma\ldots E_\sigma w_n E_\sigma v$, which implies that $uE_\sigma v$ by transitivity of $E_\sigma$. Together with C2, it implies that $\uparrow \sigma = \set{\delta \in T\mid uE_\delta v}$. From our assumption it follows that $\sigma\leq \tau$ and $\sigma\leq \tau'$. So $\sigma\leq \tau\cdot\tau'$, which implies that $uE_{\tau\cdot\tau'}v$.

    $E_{\sigma+\sigma'} \subseteq (E_\sigma\cup E_{\sigma'})^\ast$: This follows from S1.

    $E_{\sigma+\sigma'} \supseteq (E_\sigma\cup E_{\sigma'})^\ast$: Assume that $u (E_\sigma\cup E_{\sigma'})^\ast v$. By C2, $u (E_{\sigma+\sigma'})^\ast v$. By transitivity of $E_{\sigma+\sigma'}$, it follows that $uE_{\sigma+\sigma'} v$.
    
    ($E_\tau$ is an equivalence relation): obvious.
\end{proof}

\begin{lemma}[Existence Lemma for perfect network]\label{lem:existencePerfect}
    Let $\mathcal{P}_\mathcal{N}$ be a perfect network generated from a good network $\mathcal{N}$. For all formulas $\Ldiamond{\tau}\psi$ in $\neg \mathrm{FL}(\Sigma)$ and $s\in N$, $\Ldiamond{\tau}\psi\in \nu(s)$ iff there is a $s'\in N$ such that $sE_\tau s'$ and $\psi\in \nu(s')$.
\end{lemma}
\begin{proof}
    The proof is similar to that of Lemma \ref{lem:existenceTFCM}.
\end{proof}

Next, we identify those defects we may encounter during our construction of a perfect network, after which we show that they can all be repaired.

\begin{definition}
    Let $\mathcal{N} = (N,\set{R_\tau}_{\tau\in T},\nu)$ be a network. An S1-defect of $\mathcal{N}$ consists of a pair of nodes $(s,s')\in R_{\sigma+\sigma'}$ for which no path from $s$ to $s'$ exists such that each step is $R_\sigma$ or $R_{\sigma'}$. An S2-defect consists of a node $s$ and a formula $\Ldiamond{\tau}\psi\in\nu(s)$ for which there is no $s'\in N$ such that $sR_\tau s'$ and $\psi\in \nu(s')$.
\end{definition}

\begin{definition}
    Let $\mathcal{N}_0 = (N_0,\set{R^0_\tau}_{\tau\in T},\nu_0)$ and $\mathcal{N}_1 = (N_1,\set{R^1_\tau}_{\tau\in T},\nu_1)$ be two networks. We say that $\mathcal{N}_1$ extends $\mathcal{N}_0$ (notation $\mathcal{N}_1\rhd \mathcal{N}_0$) if $\mathfrak{F}_{\mathcal{N}_0}$  is a subframe of $\mathfrak{F}_{\mathcal{N}_1}$ and $\nu_0$ agrees with $\nu_1$ on $N_0$.
\end{definition}

\begin{lemma}[Repair lemma]
    For any defect of a finite, coherent network $\mathcal{N}$ there is a finite, coherent $\mathcal{N}'\rhd \mathcal{N}$ lacking this defect.
\end{lemma}
\begin{proof}

We first deal with the case when the defect is S2 type. Given the defect pair $(s,\Ldiamond{\tau}\psi)$,  choose a new point which is not in $N$ and let $B$ be an atom such that $\nu(s)S^\Sigma_\tau B$ in the FCM for $\Sigma$ (such a $B$ exists by the Existence Lemma for FCM). Define $\mathcal{N}' = (N',\set{R'_\tau}_{\tau\in T},\nu')$ as follows:
\begin{itemize}
    \item $N' := N\cup \set{s'}$;
    \item $R'_\sigma := R_\sigma$ for all $\sigma\ngeq \tau$ and $R'_\sigma := R_\sigma\cup \set{(s,s'),(s',s)}$ for all $\sigma\geq \tau$;
    \item $\nu' = \nu\cup \set{(s',B)}$.
\end{itemize}

We first prove that $\mathcal{N}'$ satisfies (C3). 
If there is a path from $t$ to $s'$ in $\mathcal{N}'$, the path must pass through $s$. Because $\uparrow \tau =\set{\delta\mid sR_\delta s'}$ and there is $\sigma\in T$ such that 
$\uparrow \sigma = \set{\delta \in T\mid \text{ there is a $\delta$-path from $t$ to $s$ in }\mathcal{N}'}$, by Proposition \ref{prop:pathFact}.(3), it follows that $\uparrow (\sigma+\tau) = \set{\delta \in T\mid \text{ there is a $\delta$-path from $t$ to $s'$ in }\mathcal{N}'}$.

The fact that everytime we add a pair into a relation $R_\tau$ we also add the pair into $R_\sigma$ for $\sigma\geq \tau$ makes sure that (C2) still holds in $\mathcal{N}'$. The reason for (C1) holding in $\mathcal{N}'$ is that all pairs we add into a relation $R_\tau$ are based on the relation $S^\Sigma_\tau$ between their lables and $S^\Sigma_\tau$ is symmetric. (C0) is guaranteed by the way we add pairs into each relation.

Next, we deal with the case when the defect is S1 type. Given the defect triple $(u,\sigma+\sigma',v)$, because $\mathcal{N}$ is coherent, we know that there is a foundational bridge $\sigma_0$ between $u$ and $v$ such that $\sigma+\sigma'\in \uparrow \sigma_0$ and $\sigma_0 \nleq \sigma$ and $\sigma_0\nleq \sigma'$.

We will construct a path $s_0 R_{\sigma_1} s_1 R_{\sigma_2}\ldots R_{\sigma_n} s_n$ so that $s_0 = u$, $s_n = v$, $\sigma_{2i+1} = \sigma$ and $\sigma_{2i} = \sigma'$,  each $s_i$ is labeled by an atom $B_i$ and $B_0 S^\Sigma_{\sigma_1} B_1 S^\Sigma_{\sigma_2} \ldots S^\Sigma_{\sigma_n} B_n$ holds in the FCM for $\Sigma$ where $n> 3$ is even.

Since $\nu(u) S^\Sigma_{\sigma +\sigma'}\nu(v)$ in FCM and $S^\Sigma_{\sigma +\sigma'}\subseteq (S^\Sigma_\sigma\cup S^\Sigma_{\sigma'})^\ast$, so there must be a path from $\nu(u)$ to $\nu(v)$ in the FCM each step  of which is  $S^\Sigma_\sigma$ or $S^\Sigma_{\sigma'}$. We take such a path, from which we can construct a walk from $\nu(u)$ to $\nu(v)$, $B_0 S^\Sigma_{\sigma_1} B_1 S^\Sigma_{\sigma_2} \ldots S^\Sigma_{\sigma_n} B_n$ where $\nu(u) = B_0$ $\nu(v) = B_n$, $n>3$, $\sigma_{2i+1} = \sigma$, $\sigma_{2i} = \sigma'$ and $n$ is even by repeating some nodes in the path using the reflexivity of the relations. Define $\mathcal{N}' = (N',\set{R'_\tau}_{\tau\in T},\nu')$ as follows:
\begin{itemize}
    \item $N' := N\cup \set{s_1,s_2\ldots,s_{n-1}}$;
    \item let $s_0 = u$ and $s_n = v$: 
        \begin{itemize}
            \item $R'_\tau := R_\tau$ for all $\tau\in T$ satisfying $\tau\ngeq \sigma$ and $\tau\ngeq \sigma'$; and
            \item $R'_\tau := R_\tau\cup \set{(s_i,s_{i+1}), (s_{i+1},s_i)\mid n>i\geq 0 \text{ is even}}$ for all $\tau\geq \sigma$; and;
            \item $R'_\tau := R_\tau\cup \set{(s_i,s_{i+1}), (s_{i+1},s_i)\mid n>i\geq 0 \text{ is odd}}$ for all $\tau\geq \sigma'$; 
        \end{itemize}
    \item $\nu' = \nu\cup \set{(s_i,B_i)\mid 1\leq i\leq n-1}$.
\end{itemize}

Conditions (C0), (C1) and (C2) of $\mathcal{N}'$ are proved in a similar way to the previous case. 

For (C3), given a path from $s$ to $s'$. Consider three cases. 

In the first case, the starting point $s$ is in $\set{s_1,s_2\ldots,s_{n-1}}$, say  $s = s_i$, and the ending point $s'$ is not in $\set{s_1,s_2\ldots,s_{n-1}}$. 

One of the following cases must hold: 
\begin{itemize}
    \item there is a path from $v$ to $s'$ in $\mathcal{N}$ which does not go via $u$;
    \item there is a path from $u$ to $s'$ in $\mathcal{N}$ which does not go via $v$.
\end{itemize}
Without loss of generality, we assume that there is a path from $v$ to $s'$ in $\mathcal{N}$ which does not go via $u$. It follows that there is $\delta\in T$ such that 
$\uparrow \delta = \set{\tau \in T\mid \text{ there is a $\tau$-path from $v$ to $s'$ in }\mathcal{N}}$.

Next, we consider all paths from $s_i$ to $v$ in $\mathcal{N}'$. All paths from $s_i$ to $v$ either goes via $u$ or does not go via $u$. 
Let 
$$u^+ = \set{\tau\in T\mid \text{ there is a $\tau$-path from $s_i$ to $v$ via $u$ in }\mathcal{N}'}$$ and 
$$
u^- = \set{\tau\in T\mid \text{ there is a $\tau$-path from $s_i$ to $v$ not via $u$ in }\mathcal{N}'}.$$
So $\set{\tau \in T\mid \text{ there is a $\tau$-path from $s_i$ to $v$ in }\mathcal{N}'} = u^+\cup u^-$.

Because there is only one path from  $s_i$ to $v$ in $\mathcal{N}'$ which does not go via $u$, that is, 
$$s_iR'_{\sigma_{i+1}} s_{i+1} R'_{\sigma_{i+2}}\ldots R'_\tau s_{n-1}R'_{\sigma_n} v\enspace .$$
It follows that $u^-$ equals $\uparrow (\sigma+\sigma')$ or $\uparrow \sigma'$, depending on whether $i<n-1$ or $i=n-1$. 
There is only one path from $s_i$ to $u$ in  $\mathcal{N}'$ which does not go via $v$ in $\mathcal{N}'$, that is, 
$uR'_{\sigma_1} s_1 R'_{\sigma_2}\ldots R'_\tau s_{i-1}R'_{\sigma_i} s_i$.
It follows that $\set{\tau\in T\mid \text{ there is a $\tau$-path from $s_i$ to $u$ not via $v$ in }\mathcal{N}'}$ equals $\uparrow (\sigma+\sigma')$ or $\uparrow \sigma$, depending on whether $i>1$ or $i=1$.

When $i = n-1$, $\uparrow (\sigma+\sigma') = \set{\tau\in T\mid \text{ there is a $\tau$-path from $s_i$ to $u$ not via $v$ in }\mathcal{N}'}$. Together with the fact that $R_{\sigma_0}$ is a foundational bridge from $u$ to $v$ in $\mathcal{N}$, it follows that $u^+ = \uparrow (\sigma+\sigma'+\sigma_0).$
Since $\sigma_0\leq \sigma+\sigma'$, 
$u^+ = \uparrow (\sigma+\sigma')$.
Combining with the facts that $u^-$ equals $\uparrow \sigma'$ and $\sigma'\leq \sigma+\sigma'$, it follows that 
$\uparrow \sigma' = \set{\tau\in T\mid \text{ there is a $\tau$-path from $s_i$ to $v$ in }\mathcal{N}'}$, 
that is,  the foundational bridge from $s_i$ to $v$ in $\mathcal{N}'$ is $R'_{\sigma'}$.

When $1<i<n-1$, it is still the case that $u^+ = \uparrow (\sigma+\sigma')$ using a similar argument in the case $i+1 =n$. But $u^-$ equals $\uparrow (\sigma+\sigma')$ in this case. Thus it follows that 
$$\uparrow (\sigma+\sigma') = \set{\tau\in T\mid \text{ there is a $\tau$-path from $s_i$ to $v$ in }\mathcal{N}'}$$
that is,  the foundational bridge from $s_i$ to $v$ in $\mathcal{N}'$ is $R'_{\sigma+\sigma'}$.

When $i = 1$, $\uparrow \sigma = \set{\tau\in T\mid \text{ there is a $\tau$-path from $s_i$ to $u$ not via $v$ in }\mathcal{N}'}$. Together with the fact that $R_{\sigma_0}$ is the foundational bridge from $u$ to $v$ in $\mathcal{N}$, it follows that 
$u^+ = \uparrow (\sigma_0+\sigma)$.
Since $u^-$ equals $\uparrow (\sigma+\sigma')$ in this case and $\sigma_0+\sigma\leq \sigma+\sigma'$, it follows that
$$\uparrow (\sigma_0+\sigma) = \set{\tau\in T\mid \text{ there is a $\tau$-path from $s_i$ to $v$ in }\mathcal{N}'}$$
that is,  the foundational bridge from $s_i$ to $v$ in $\mathcal{N}'$ is $R'_{\sigma_0+\sigma}$. 

Therefore,  the foundational bridge from $s_i$ to $s$ in $\mathcal{N}'$ is either $R'_{\sigma'+\delta}$, $R'_{\sigma+\sigma'+\delta}$ or $R'_{\sigma_0+\sigma+\delta}$.

The second case where the ending point $s'$ is in $\set{s_1,s_2\ldots,s_{n-1}}$ but the starting point $s$ is not in $\set{s_1,s_2\ldots,s_{n-1}}$. This case is symmetric to the first case, so it follows by C0.

In the third case, $s$ and $s'$ are both in $\set{s_1,s_2\ldots,s_{n-1}}$, say $s=s_i$ and $s'=s_j$ and $i\leq j$. The proof is similar to the one we use in the first case for showing that there is a foundational bridge between $s$ and $v$ in $\mathcal{N}'$. In this case, there are two ways to go from $s_i$ to $s_j$. The first one is $s_iR_{\sigma_i+1}\ldots R_{\sigma_j}s_j$ which does not go via $u$ or $v$. It is easy to see that the set of all possible paths via this way is either $\uparrow (\sigma+\sigma')$, $\uparrow \sigma$ or $\uparrow \sigma'$.  The other one goes via $u$ and $v$. The set of all possible paths is $\uparrow (\sigma +\sigma') $. Therefore, the foundational bridge from $s_i$ to $s_j$ in $\mathcal{N}'$ is either $R'_{\sigma+\sigma'}$, $R'_{\sigma}$ or $R'_{\sigma'}$.

In the last case, neither $s$ nor $s'$ is in $\set{s_1,s_2\ldots,s_{n-1}}$. If there is no path from $s$ to $s'$ in $\mathcal{N}$ going via both $u$ and $v$, the foundational bridge from $s$ to $s'$ in $\mathcal{N}$ is the foundational bridge from $s$ to $s'$ in $\mathcal{N}'$.

Assume that there are paths from $s$ to $s'$ in $\mathcal{N}$ going via both $u$ and $v$.  
There is $\delta\in T$ such that $\uparrow \delta = \set{\tau \in T\mid \text{ there is a $\tau$-path from $s$ to $s'$ in }\mathcal{N}}$. 

Obviously, $\uparrow \delta\subseteq \set{\tau \in T\mid \text{ there is a $\tau$-path from $s_i$ to $s'$ in }\mathcal{N}'}$.

For the other direction, take any $\tau$-path from $s$ to $s'$ in $\mathcal{N}'$, 
$sR'_\tau w_1 R'_\tau\ldots R'_\tau w_mR'_\tau s'$.
If the path goes via $s_1,\ldots,s_{n-1}$, then $\sigma\leq \tau$ and $\sigma'\leq \tau$. So $\sigma+\sigma'\leq \tau$ and we can delete $s_1,\ldots,s_{m-1}$ from the path and connect $u$ and $v$ by a $\tau$-path in $\mathcal{N}$, becauase $\sigma_0\leq \sigma+\sigma'\leq \tau$. This implies that there is a path in $\mathcal{N}$:
$sR_\tau w_1 R_\tau\ldots uR_\tau\ldots R_\tau v\ldots  R_\tau w_mR_\tau s'$
So $\tau\in \uparrow \delta$. 

Therefore, $\uparrow \delta = \set{\tau \in T\mid \text{ there is a $\tau$-path from $s$ to $s'$ in }\mathcal{N}'}$.
\end{proof}

\begin{lemma}\label{lem:goodCons}
    For an atom $A\in At(\Sigma)$, there is a good network where one of its nodes is labelled by $A$. 
\end{lemma}
\begin{proof}
    The proof is similar to that of Theorem 4.65 in \cite{Blackburn2001}.

    Choose some set $S = \set{s^i\mid i\in \omega}$ and enumerate the set of potential defects (that is, the union of the sets $S\times T\times T\times S$ and $S\times T\times \mathcal{L}$). Given an atom $A\in At(\Sigma)$, let $\mathcal{N}_0$ be the network $(\set{s^0},\emptyset,(s^0,A))$. Trivially, $\mathcal{N}_0$ is a finite, coherent network.

    Let $n\geq 0$ and suppose $\mathcal{N}_n$ is a finite, coherent network. If there is no defect in $\mathcal{N}_n$, let $\mathcal{N}_m = \mathcal{N}_n$ for all $m>n$. If there are defects in $\mathcal{N}_n$, let $D$ be the defect of $\mathcal{N}_n$ that is minimal in our enumeration.  Form $\mathcal{N}_{n+1}$ by repairing the defect $D$ as described in the proof of the Repair Lemma. Observe that $D$ will not be a defect of any network extending $\mathcal{N}_{n+1}$.

    Let $\mathcal{N} = (N,\set{R_\tau}_{\tau\in T},\nu)$ be given by 
    $N = \bigcup_{n\in \omega} N_n,\quad R_\tau = \bigcup_{n\in \omega} R^n_\tau \text{ for all }\tau\in T,\quad \nu = \bigcup_{n\in \omega} \nu_n$.
    It is a coherent network since our repair preserves all conditions of coherency. Moreover, $\mathcal{N}$ is saturated. Suppose otherwise, that is,there are some defects in $\mathcal{N}$. From these defects, take the minimal one in the enumeration, say $D_k$. It must be a defect in some approximation of $\mathcal{N}$, say $\mathcal{N}_n$. Although $D_k$ is not necessarily the minimal defect in $\mathcal{N}_n$, there are finite defects before $D_k$ in $\mathcal{N}_n$ to be repaired afterwards. After finite many steps of repair, $D_k$ becomes the minimal defect to be repaired. So $D_k$ is repaired at some stage of the process of repair, contradiction.   
\end{proof}

\subsection{Proving the completeness of $S5\mathsf{LCDK}$}

For any finite consistent set of formulas $\Sigma$, extend it to an atom $A$ over $\Sigma$ and construct a good network $\mathcal{N}$ according to Lemma \ref{lem:goodCons}. Construct the perfect network $\mathcal{P}_\mathcal{N}$ based on the good network $\mathcal{N}$. Let $\mathfrak{F}_\mathcal{P}$ be the frame underlying the perfect network $\mathcal{P}$. Finally, prove the truth lemma for the induced model $\mathfrak{I}_\mathcal{P} = (\mathfrak{F}_\mathcal{P},V_\mathcal{N})$ using the Existence Lemma for perfect network (Lemma \ref{lem:existencePerfect}).
Because the frame underlying a perfect network is S5-regular, we have proved the completeness of $S5\mathsf{LCDK}$ with respect to the class of S5-regular frames.

\subsection{Proving the completeness of $S4\mathsf{LCDK}$}

In the case of S4, Lemma \ref{lem:reflexSymm} should be changed to 
\begin{lemma}
    The canonical relations are reflexive.
\end{lemma}
Proposition \ref{prop:AtauEBtau} should be changed to 
\begin{proposition}
    $AS^\Sigma_\tau B \text{ implies that } B_\tau \subseteq A_\tau$
\end{proposition}
Lemma \ref{lem:existenceTFCM} makes use of Proposition \ref{prop:AtauEBtau}. Its proof in the S4 case should be changed accordingly.

The condition C0 in the definition of coherent networks should be removed. The proofs in Section \ref{sec:stepBystep} which are based on the condition C0 need to be modified accordingly. In the proof of the repair lemma, when constructing new networks, we keep the forward pairs $(s,s')$, $(s_i,s_{i+1})$ and get rid of the backward pairs $(s',s)$, $(s_{i+1},s_i)$. The proof that our repair of the two types of defects preserves all conditions of coherency becomes simpler, because the paths from $s_1$ to $s_{n-1}$ in $\mathcal{N}'$ are all one-way paths. For S1 defect, the original proof still works. For S2 defect, we only need to consider the first, the third and the fourth cases we consider in the original proof. For the first case, the only way to reach $s'$ from $s_i$ is via $v$ and all paths from $s_i$ to $v$ go only via $s_k$ where $i\leq k <n$. For the third case, all path from $s_i$ to $s_j$ go only via $s_k$ where $i\leq k\leq j$. For the fourth case, the orignal proof still works.

\paragraph{Acknowledgements} I'd like to thank Alexandru Baltag for his suggestions and comments on several early versions of this paper, which help clarify some critical components of the logic. Many thanks to the three anonymous reviewers for their useful suggestions on improving the presentation of the paper. 

\nocite{*}
\bibliographystyle{eptcs}
\bibliography{LCDK}
\end{document}